\documentclass[conference]{IEEEtran}
\usepackage{amsfonts,amsmath,amssymb}
\usepackage{lscape,graphicx}
\newtheorem{theorem}{\indent Theorem}[section]
\newtheorem{lemma}[theorem]{\indent Lemma}
\newtheorem{corollary}[theorem]{\indent Corollary}

\newtheorem{proposition}[theorem]{\indent Proposition}

\newtheorem{EXAMPLE}{\indent Example}[section]
\newtheorem{definition}{\indent Definition}[section]

\newenvironment{example}{\begin{EXAMPLE}\rm}{\rm\end{EXAMPLE}}

\newcommand{\define}{\stackrel{\mbox{\tiny $\triangle$}}{=}}
\newcommand{\AWGNC}{\mbox{\tiny AWGNC}}
\newcommand{\BSC}{\mbox{\tiny BSC}}
\newcommand{\BEC}{\mbox{\tiny BEC}}
\newcommand{\maxfrac}{\mbox{\tiny max-frac}}
\newcommand{\dual}{^{\perp}}

\newcommand{\code}{\mathcal{C}}

\newcommand{\cL}{{\mbox{\boldmath $L$}}}
\newcommand{\cV}{{\mathcal{V}}}
\newcommand{\cHsmall}{{\mbox{\scriptsize \boldmath $H$}}}
\newcommand{\cH}{{\mbox{\boldmath $H$}}}

\newcommand{\cI}{{\mathcal{I}}}
\newcommand{\cJ}{{\mathcal{J}}}

\newcommand{\cK}{{\mathcal{K}}}

\newcommand{\weight}{{\mathsf{w}}}

\newcommand{\ff}{{\mathbb{F}}}

\newcommand\rr{{\mathbb R}}

\newcommand\nn{{\mathbb N}}

\newcommand{\bldc}{{\mbox{\boldmath $c$}}}

\newcommand{\bldx}{{\mbox{\boldmath $x$}}}
\newcommand{\bldxx}{{\mbox{\scriptsize \boldmath $x$}}}

\newcommand{\smallzeros}{{\mbox{\scriptsize \boldmath $0$}}}

\newcommand{\bldzero}{{\mbox{\boldmath $0$}}}

\newcommand{\entropy}{\mathsf{H}}

    \def\squarebox#1{\hbox to #1{\hfill\vbox to #1{\vfill}}}

\newcommand{\mat}[1]{\left[\ \ \begin{matrix}#1\end{matrix}\;\ \right]}
\newcommand{\zo}{\mbox{\footnotesize 0 \normalsize}\!\!}
\newcommand{\ze}{\mbox{\footnotesize $\mathbf{1}$ \normalsize}\!\!}

\newlength{\Algwidth}

\title{On the Pseudocodeword Redundancy}

\author{Jens Zumbr\"agel$^*$, Mark F. Flanagan$^*$, and Vitaly
  Skachek$^\dagger$\medskip\\%
  {\normalsize $^*$Claude Shannon Institute, University College
    Dublin, Belfield, Dublin 4, Ireland}\\
  {\normalsize
    (e-mail: jens.zumbragel@ucd.ie, mark.flanagan@ieee.org).}\\
  {\normalsize $^\dagger$Division of Mathematical Sciences, School
    of Physical and Mathematical Sciences, Nanyang Technological
    University,}\\
  {\normalsize 21 Nanyang Link, 637371 Singapore (e-mail:
    vitaly.skachek@ntu.edu.sg). }  }

\begin{document}


\maketitle

\begin{abstract}
  We define the AWGNC, BSC, and max-fractional \emph{pseudocodeword
    redundancy} $\rho(\code)$ of a code $\code$ as the smallest number
  of rows in a parity-check matrix such that the corresponding minimum
  pseudoweight is equal to the minimum Hamming distance of $\code$. We
  show that most codes do not have a finite $\rho(\code)$. We also
  provide bounds on the pseudocodeword redundancy for some families of
  codes, including codes based on designs.
\end{abstract}


\section{Introduction}

Pseudocodewords represent the intrinsic mechanism of failure of
binary linear codes under linear-programming (LP) or message-passing
(MP) decoding.  In~\cite{FKKR}, the pseudocodeword effective Euclidean
weight, or \emph{pseudoweight}, was associated with any
pseudocodeword. This concept of pseudoweight was shown to play an
analogous role to that of the signal Euclidean distance (AWGNC) or
Hamming distance (BSC) in the ML decoding scenario. The minimum
pseudoweight of the code $\code$ with respect to a parity-check matrix
$\cH$ is defined as the minimum over all pseudoweights of nonzero
pseudocodewords; this may be considered as a first-order measure of
decoder error-correcting performance for LP or MP decoding. Typically,
a lower minimum pseudoweight corresponds to a higher probability of
decoding error. Another measure closely related to BSC pseudoweight is
the max-fractional weight (pseudoweight).  It serves as a lower bound
on both AWGNC and BSC pseudoweights.

In order to minimise the decoding error probability under LP (or MP)
decoding, one might want to select a matrix $\cH$ which maximises the
minimum pseudoweight of the code for the given channel.  However,
generally it is not clear how this goal may be achieved. Adding
redundant rows to the parity-check matrix introduces additional
constraints on the so-called \emph{fundamental cone}, and may thus
increase the minimum pseudoweight. However, such additions increase
the decoding complexity under MP decoding, especially since linear
combinations of low-density rows may not yield a low-density result.
On the other hand, there exist classes of codes for which sparse
parity-check matrices exist with many redundant rows,
e.g. \cite{kou_lin_fossorier}.

For the AWGNC, BEC, BSC pseudoweights, and max-fractional weight,
define $\rho_{\AWGNC}(\code)$, $\rho_{\BEC}(\code)$,
$\rho_{\BSC}(\code)$, and $\rho_{\maxfrac}(\code)$, respectively, to be
the minimum number of rows in any parity-check matrix $\cH$ such that
the minimum pseudoweight of $\code$ with respect to this matrix is
equal to the code's minimum distance $D$. For the sake of simplicity,
we sometimes use the notation $\rho(\code)$ when the type of channel
is clear from the context.  The value $\rho(\code)$ is called the
(AWGNC, BEC, BSC, max-fractional) \emph{pseudocodeword redundancy} (or
pseudoredundancy) of $\code$. If for the code $\code$ there exists no
such matrix $\cH$, we say that the pseudoredundancy is infinite.

The BEC pseudocodeword redundancy was studied
in~\cite{Schwartz-Vardy}, where it was shown that for any binary
linear code $\code$ there exists a parity-check matrix $\cH$ such that
the minimum pseudoweight with respect to this $\cH$ is equal to $D$,
and therefore the BEC pseudocodeword redundancy is finite for all
codes. The authors also presented some bounds on $\rho_{\BEC}(\code)$
for general linear codes, and for some specific families of codes.

In this work, we address the analogous problem for the AWGNC, BSC, and
max-fractional pseudoweight. We show that for most codes there exists
no $\cH$ such that the minimum pseudoweight (with respect to $\cH$) is
equal to $D$, and therefore the AWGNC, BSC, and max-fractional
pseudocodeword redundancy (as defined above) is infinite for most
codes. For some code families for which the pseudoredundancy is
finite, we provide upper bounds on its value.


\section{General Settings}

Let $\code$ be a code of length $n \in \nn$ over the binary field
$\ff_2$, defined by
\begin{equation}\label{eq:code_definition}
  \code = \ker \cH
  = \{ \bldc \in \ff_2^n \; : \; \cH \bldc^T = \bldzero^T \}
\end{equation}
where $\cH$ is an $m \times n$ \emph{parity-check matrix} of the code
$\code$. Obviously, the code $\code$ may admit more than one
parity-check matrix, and all the codewords form a linear vector space
of dimension $k \ge n-m$. We say that $k$ is the \emph{dimension} of
the code $\code$.  The \emph{rate} of the code $\code$ is defined as
$R(\code) = k/n$ and is equal to the number of information bits per
coded bit. We denote by $D$ the minimum Hamming distance (also called
the minimum distance) of $\code$. The code $\code$ may then be
referred to as an $[n,k,D]$ linear code over $\ff_2$.

Denote the set of column indices and the set of row indices of $\cH$
by $\cI = \{1, 2, \dots, n \}$ and $\cJ = \{1, 2, \dots, m \}$,
respectively. For $j \in \cJ$, we denote $\cI_j \define \{ i \in \cI
\; : \; H_{j,i} \neq 0 \}$, and for $i \in \cI$, we denote $\cJ_i
\define \{ j \in \cJ \; : \; H_{j,i} \neq 0 \}$. The \emph{fundamental
  cone} of $\cH$, denoted $\cK(\cH)$, is defined as the set of vectors
$\bldx \in \rr^n$ that satisfy
\begin{equation}\label{eq:polytope-inequality-1}
  \forall j \in \cJ, \; \forall \ell \in \cI_j  \; : 
  \; x_\ell \le \sum_{i \in \cI_j \backslash \{ \ell \}} x_i \; ,
\end{equation}
\begin{equation}\label{eq:polytope-inequality-2}
  \forall i \in \cI \; : \; x_i \ge 0 \; .
\end{equation}

The vectors $\bldx\in\rr^n$
satisfying~(\ref{eq:polytope-inequality-1})
and~(\ref{eq:polytope-inequality-2}) are called \emph{pseudocodewords}
of $\code$ with respect to the parity-check matrix $\cH$.  Note that
the fundamental cone $\cK(\cH)$ depends on the parity-check matrix
$\cH$ rather than on the code $\code$ itself.  At the same time, the
fundamental cone is independent of the underlying communication
channel.

The BEC, AWGNC, BSC pseudoweights, and max-fractional weight of a
nonzero pseudocodeword $\bldx \in \cK(\cH)$ were defined
in~\cite{FKKR} and~\cite{KV-long-paper} as follows:
\begin{align*}
  \weight_{\BEC} (\bldx) & \,\define\, 
  \left| \mbox{supp} ( \bldx ) \right| \; , \\
  \weight_{\AWGNC} (\bldx) & \,\define\,  
  \frac{\left( \sum_{i \in \cI} x_i \right)^2}{\sum_{i \in \cI} x_i^2}\;.
\end{align*}
Let $\bldx'$ be a vector in $\rr^n$ with the same components as
$\bldx$ but in non-increasing order.  For $i-1 < \xi \le i$, where $1
\le i \le n$, let $\phi(\xi) \stackrel{\triangle}{=} x'_i$. Define
$\Phi(\xi) \define \int_{0}^{\xi} \phi(\xi') \; d \xi'$ and \[
\weight_{\BSC}(\bldx) \define 2\, \Phi^{-1} ( \Phi(n)/2 ) \; . \] 
Finally, the max-fractional weight of $\bldx$ is defined as
\[\weight_{\maxfrac} (\bldx) \,\define\,
\frac{\sum_{i \in \cI} x_i}{\max_{i \in \cI} x_i} \; .\]

We define the BEC \emph{minimum pseudoweight} of the code $\code$ with
respect to the parity-check matrix $\cH$ as
\[ \weight_{\BEC}^{\min} (\cH) \,\define\, \min_{\bldxx \in \cK(\cHsmall)
  \backslash \{ \smallzeros \} } \weight_{\BEC} (\bldx) \; .\] The
quantities $\weight_{\AWGNC}^{\min} (\cH) $, $\weight_{\BSC}^{\min}
(\cH) $ and $\weight_{\maxfrac}^{\min} (\cH)$ are defined
similarly. Note that all four minimum pseudoweights are upper bounded
by $D$, the code's minimum distance.

Then we define the BEC \emph{pseudocodeword redundancy} of the code
$\code$ as
\[ \rho_{\BEC}(\code) \,\define\, \inf\{\#\text{rows}(\cH) \mid
\ker\cH=\code\,,\, \weight_{\BEC}^{\min}(\cH)=D\} \: ,\] where
$\inf\varnothing\define\infty$, and similarly we define the
pseudocodeword redundancies $\rho_{\AWGNC}(\code)$,
$\rho_{\BSC}(\code)$ and $\rho_{\maxfrac}(\code)$ for the AWGNC and
BSC pseudoweights, and the max-fractional weight.  We remark that all
pseudocodeword redundancies satisfy $\rho(\code) \ge n-k$.


\section{Basic Connections} 

The next lemma is taken from~\cite{KV-long-paper}.

\begin{lemma}\label{lemma:relations}
  Let $\code$ be a binary linear code with the parity-check matrix
  $\cH$.  Then,
  \begin{eqnarray*}
    &\weight_{\maxfrac}^{\min} (\cH) \; \le \; 
    \weight_{\AWGNC}^{\min} (\cH) \; \le \; \weight_{\BEC}^{\min} (\cH) \; , \\
    &\weight_{\maxfrac}^{\min} (\cH) \; \le \; 
    \weight_{\BSC}^{\min} (\cH) \; \le \; \weight_{\BEC}^{\min} (\cH) \; . 
  \end{eqnarray*}
\end{lemma}

The following theorem is a straightforward corollary to
Lemma~\ref{lemma:relations}.

\begin{theorem}\label{thrm:pseudoredundancies}
  Let $\code$ be a binary linear code.  Then,
  \begin{eqnarray*}
    &\rho_{\maxfrac} (\code) \; \ge \; \rho_{\AWGNC} (\code) \; \ge \; 
    \rho_{\BEC} (\code) \; , \\
    &\rho_{\maxfrac} (\code) \; \ge \; \rho_{\BSC} (\code) \; \ge \; 
    \rho_{\BEC} (\code) \; .
  \end{eqnarray*}
\end{theorem}

We note that for \emph{geometrically perfect} codes, a class of codes
defined and characterised in~\cite{Kashyap_decomp}, all four
pseudocodeword redundancies are finite.


\section{Pseudoredundancy of Random Codes} 

We begin with the following lemma. 

\begin{lemma}\label{lemma:awgn}
  For the binary linear code $\code$ of length $n$, let $d$ be the
  minimum distance of the dual code.  Then, the minimum AWGNC
  pseudoweight of $\code$ (with respect to any parity-check matrix
  $\cH$) satisfies
  \begin{equation}\label{eq:pseudo-distance-bound-random}
    \weight_{\AWGNC}^{\min} \le \frac{(n + d - 2)^2}{(d-1)^2 + (n-1)} \; . 
  \end{equation}
\end{lemma}

\begin{proof} 
  Consider the pseudocodeword $\bldx = (x_1, x_2, \dots, x_n) \define
  (d-1, 1, \dots, 1)$.  Since $d$ is the minimum distance of the dual
  code, every row in $\cH$ has weight at least $d$.  Therefore, all
  inequalities~(\ref{eq:polytope-inequality-1})
  and~(\ref{eq:polytope-inequality-2}) are satisfied for this $\bldx$,
  and so it is indeed a legal pseudocodeword. Finally, observe that
  the AWGNC pseudoweight of $\bldx$ is given by the right-hand side of
  (\ref{eq:pseudo-distance-bound-random}).
\end{proof}

Next, we take a random binary linear code $\code$ of a fixed rate $R$
and arbitrary length $n$ (for $n \rightarrow \infty$).  It is well
known that the relative minimum distance $\delta = D/n$ of $\code$
attains, with probability approaching $1$ as $n \rightarrow \infty$,
for any fixed small $\epsilon > 0$, the Gilbert-Varshamov bound
\[ \delta \ge \entropy_2^{-1} (1 - R) - \epsilon \; , \] where
$\entropy^{-1}_2(\cdot)$ is the inverse of the binary entropy function
$\entropy_2 (p) = - p \log_2 p - (1 - p) \log_2 (1 - p)$
(see~\cite[Theorems 4.4, 4.5, and 4.10]{Roth-book} for details).

The dual code $\code\dual$ of $\code$ can be viewed as a random code
also, and so with high probability the rate $R\dual=1-R$ and the
relative minimum distance $\delta\dual = d/n$ of the dual code attain
the Gilbert-Varshamov bound
\[ \delta\dual \ge \mu \define \entropy^{-1}_2 (1 - R\dual) - \epsilon
= \entropy^{-1}_2(R) - \epsilon \; , \] 

Note that~(\ref{eq:pseudo-distance-bound-random}) may be written in
terms of the relative minimum distance $\delta\dual$ of the dual code
as follows:
\begin{equation}
  \weight_{\AWGNC}^{\min} \le 
  \frac{(1 + \delta\dual - 2/n)^2}{(\delta\dual- 1/n)^2 + (1/n-1/n^2)} \; .
  \label{eq:pseudo-distance-new}
\end{equation}

For large $n$, the minimum pseudoweight of the code $\code\dual$ is
bounded from above by the constant $(1 + 1/\delta\dual)^2 + \epsilon'
\le (1+1/\mu)^2 + \epsilon'$ for some small $\epsilon' > 0$ --- this
bound does not depend on $n$. On the other hand, $\code$ is a random
code and so its minimum distance satisfies the Gilbert-Varshamov
bound, namely
\[ D\ge \left( \entropy^{-1}_2 (1 - R) - \epsilon \right) \cdot n \; ,\]
which increases linearly with $n$ for a fixed $R$. 

We obtain that for any $\cH$, there is a gap between the minimum
pseudoweight and the minimum distance of a random code
$\code$. Therefore, we have the following corollary.

\begin{corollary}\label{cor:non-existance_AWGNC}
  Let $0<R<1$ be fixed.  For $n$ large enough, for a random binary
  linear code $\code$ of length $n$ and rate~$R$, there is a gap
  between the minimum AWGNC pseudoweight (with respect to any
  parity-check matrix) and the minimum distance.  Therefore, the AWGNC
  pseudoredundancy is infinite for most codes.
\end{corollary}\medskip

The following lemma is a counterpart of Lemma~\ref{lemma:awgn} for the
BSC.

\begin{lemma}\label{lemma:bsc}
  Let $\code$ be a binary linear code of length $n$, and let $d$ be the
  minimum distance of the dual code.  Then, the minimum BSC
  pseudoweight of $\code$ (with respect to any parity-check matrix
  $\cH$) satisfies
  \[ \weight_{\BSC}^{\min} \le 2 \lceil n/d \rceil \; . \]
\end{lemma}

\begin{proof}
Consider the pseudocodeword 
\[ \bldx = (x_1, x_2, \dots, x_n) \define (\underbrace{d-1, d-1,
  \dots, d-1}_\tau, \underbrace{1, 1, \dots, 1}_{n-\tau}) \; , \] 
for some positive integer $\tau$. This $\bldx$ is then a legal
pseudocodeword; since $d$ is the minimum distance of the dual code,
every row in $\cH$ has a weight of at least $d$, and so, all
inequalities~(\ref{eq:polytope-inequality-1})
and~(\ref{eq:polytope-inequality-2}) are satisfied by this $\bldx$.

If $\tau(d-1) \ge n - \tau$ then by the definition of the BSC
pseudoweight $\weight_{\BSC} (\bldx) \le 2 \tau$. This condition is
equivalent to $\tau d \ge n$. Therefore, we set $\tau = \lceil n/d
\rceil$. For the corresponding $\bldx$, the pseudoweight is less or
equal to $2\tau = 2 \lceil n/d \rceil$.
\end{proof}

Similarly to the AWGNC case, let $\code$ be a random binary code of
length $n$ and a fixed rate $R$.  The parameters $R\dual$
and~$\delta\dual$ of its dual code $\code\dual$ attain with high
probability the Gilbert-Varshamov bound $\delta\dual\ge\mu$.

From Lemma~\ref{lemma:bsc}, for all $n$, the pseudoweight of the
code~$\code\dual$ is bounded from above by 
\[2 \lceil n/d \rceil < 2/\delta\dual + 2 \le 2/\mu + 2 \; ,\]
which is a constant.  On the other hand, $\code$ is a random code and
its minimum distance also satisfies the Gilbert-Varshamov bound, so it
increases linearly with $n$.  It follows that for any $\cH$, there is
a gap between the minimum BSC pseudoweight and the minimum distance
of a random code $\code$.

\begin{corollary}\label{cor:non-existence_BSC}
  Let $0<R<1$ be fixed.  For $n$ large enough, for a random binary
  linear code $\code$ of length $n$ and rate~$R$, there is a gap
  between the minimum BSC pseudoweight (with respect to any
  parity-check matrix) and the minimum distance.  Therefore, the BSC
  pseudoredundancy is infinite for most codes.
\end{corollary}

The last corollary disproves the conjecture in~\cite{Kelley-Sridhara}
that the BSC pseudoredundancy is finite for all binary linear codes.

\begin{example}
  Consider the [23,12] Golay code having minimum distance $D=7$. The
  minimum distance of its dual code is $d=8$.  We can take a
  pseudocodeword $\bldx$ as in the proof of Lemma~\ref{lemma:bsc} with
  $\tau = 3$.  We have $\weight_{\BSC} (\bldx) \le 2 \tau = 6$, thus
  obtaining that the minimum distance is not equal to the minimum
  pseudoweight.

  Similarly, for the [24,12] extended Golay code we have $D=d=8$, and
  by taking $\tau=3$ we obtain $\weight_{\BSC} (\bldx) \le 2\tau = 6$.

  Note however that the presented techniques do not answer the
  question of whether these Golay codes have finite AWGNC
  pseudoredundancy.
\end{example}

\section{Codes Based on Designs}

\begin{definition}\label{prop:design_incidence_matrix}
  A \emph{partial $(w_c,\lambda)$ design} is a block design consisting
  of an $n$-element set $\cV$ (whose elements are called
  \emph{points}) and a collection of $m$ subsets of $\cV$ (called
  \emph{blocks}) such that every point is contained in exactly $w_c$
  blocks and every $2$-element subset of $\cV$ is contained in at most
  $\lambda$ blocks. The \emph{incidence matrix} of a design is an $m
  \times n$ matrix $\cH$ whose rows correspond to the blocks and whose
  columns correspond to the points, and satisfies $H_{j,i} = 1$ if
  block $j$ contains point~$i$, and $H_{j,i} = 0$ otherwise.

  If each block contains the same number $w_r$ of points and every
  $2$-element subset of $\cV$ is contained in exactly $\lambda$
  blocks, the design is said to be an $(n,w_r,\lambda)$ \emph{balanced
    incomplete block design} (BIBD).
\end{definition}


Note that for a BIBD we have $n w_c = m w_r$ and also
\begin{equation}\label{eq:design_constraint}
  w_c(w_r-1) = \lambda(n-1) \; ,
\end{equation}
so all other parameters may be deduced from $(n,w_r,\lambda)$; in
particular, $w_c=\frac{n-1}{w_r-1}\,\lambda$.  Note that
\cite{vasic_milenkovic} and \cite{kashyap_vardy} consider parity-check
matrices based on BIBDs; these matrices are the transpose of the
incidence matrices defined here.

We have the following general result for codes based on partial
$(w_c,\lambda)$ designs.

\begin{theorem}\label{thm:pseudoweight_bound}
  Let $\code$ be a code with parity-check matrix $\cH$, such that a
  subset of the rows of $\cH$ forms the incidence matrix for a partial
  $(w_c,\lambda)$ design. Then the minimum max-fractional weight of
  $\code$ with respect to $\cH$ is lower bounded by
  \begin{equation}
    \weight_{\maxfrac}^{\min} \ge 1 + \frac{w_c}{\lambda} \; .
    \label{eq:new_lower_bound}
  \end{equation}
  For the case of an $(n,w_r,\lambda)$ BIBD, the lower bound in
  (\ref{eq:new_lower_bound}) may also be written as $1 +
  \frac{n-1}{w_r-1}$; the alternative form follows trivially from
  (\ref{eq:design_constraint}).
\end{theorem}

\begin{proof}
  Consider the subset of the rows of $\cH$ which forms the incidence
  matrix for a partial $(w_c,\lambda)$ design.  Let $\bldx$ be a
  nonzero pseudocodeword and let $x_\ell$ be a maximal coordinate of
  $\bldx$ ($\ell \in \cI$).  For all $j \in \cJ$ such that $\ell \in
  \cI_j$, sum inequalities (\ref{eq:polytope-inequality-1}). We have
  \begin{equation*}
    w_c \cdot x_{\ell} \le \lambda \cdot 
    \sum_{i \in \cI \backslash \{ \ell \} } x_i \; ,
  \end{equation*}
  or 
  \begin{equation}\label{eq:basic_step_for_bound_bsc} 
    \left( 1 + \frac{w_c}{\lambda} \right) x_{\ell} \le 
    \sum_{i \in \cI} x_i \; . 
  \end{equation}
  The result now easily follows from the definition of
  $\weight_{\maxfrac}^{\min}$.
\end{proof}

\begin{theorem}\label{thm:pseudoweight_bound_2}
  Let $\code$ be a code with parity-check matrix $\cH$, such that a
  subset of the rows of $\cH$ forms the incidence matrix for a
  partial $(w_c,\lambda)$ design. Then,
  \begin{align*}
    \weight_{\AWGNC}^{\min} & \ge 1 + \frac{w_c}{\lambda} \; , \\
    \weight_{\BSC}^{\min} & \ge 1 + \frac{w_c}{\lambda} \; . 
  \end{align*}
\end{theorem}

The proof follows from Lemma~\ref{lemma:relations} and
Theorem~\ref{thm:pseudoweight_bound}.\medskip

Another tool for proving lower bounds of the minimum AWGNC pseudoweight
is provided by the following eigenvalue-based bound from
\cite{KV-lower-bounds}.

\begin{proposition}\label{prop:KV_bound}
  The minimum AWGNC pseudoweight for a $(w_c,w_r)$-regular
  parity-check matrix $\cH$ whose corresponding Tanner graph is
  connected is bounded below by
  \begin{equation}\label{eq:KV_bound}
    \weight_{\AWGNC}^{\min} \ge n \cdot \frac{ 2w_c - \mu_2 } 
    {\mu_1 - \mu_2 } \; ,
  \end{equation}
  where $\mu_1$ and $\mu_2$ denote the largest and second largest
  eigenvalue (respectively) of the matrix $\cL \define \cH^T \cH$,
  considered as a matrix over the real numbers.
\end{proposition}

In the case where $\cH$ is equal to the incidence matrix for an
$(n,w_r,\lambda)$ BIBD, it is easy to check that the bound of
Proposition \ref{prop:KV_bound} becomes




\[\weight^{\min}_{\AWGNC} \ge 1 + \frac{ w_c }{ \lambda } \; . \]
We conclude that in this case the bound of
Proposition~\ref{prop:KV_bound} coincides with that of Theorem
\ref{thm:pseudoweight_bound_2} (for the case of the AWGNC
only).\medskip

Note that the pseudoweight bounds of~\cite{Smarandache_PG_EG} for the
EG$(2,q)$ and PG$(2,q)$ codes for $q=2^s \ge 2$ follow from
Theorem~\ref{thm:pseudoweight_bound_2}.  We next apply the bounds of
Theorems~\ref{thm:pseudoweight_bound}
and~\ref{thm:pseudoweight_bound_2} to some other examples of codes
derived from designs.

\begin{proposition}\label{prop:Hamming}
  For $m \ge 2$, the $[2^m-1,2^m-1-m,3]$ Hamming code has AWGNC, BSC,
  and max-fractional pseudocodeword redundancies $\rho(\code) \le
  2^m-1$.
\end{proposition}

\begin{proof}
  For $m \ge 2$, consider the binary parity-check matrix $\cH$ whose
  rows are exactly the nonzero codewords of the dual code
  $\code\dual$, in this case the $[2^m-1,m,2^{m-1}]$ simplex
  code. This $\cH$ is the incidence matrix for a BIBD with parameters
  $(n,w_r,\lambda)$ = $(2^m-1, 2^{m-1},
  2^{m-2})$. Theorem~\ref{thm:pseudoweight_bound} gives
  $\weight_{\maxfrac} (\bldx) \ge 3$, leading to $\rho_{\maxfrac}
  (\code) \le 2^m-1$.  The result for AWGNC and BSC follows by
  applying Theorem~\ref{thrm:pseudoredundancies}.
\end{proof}

In the next example, we consider simplex codes. Straightforward
application of the previous reasoning does not lead to the desired
result.  However, more careful selection of the matrix $\cH$, as
described below, leads to a new bound on the pseudoredundancy.

\begin{proposition}\label{prop:Simplex}
  For $m \ge 2$, the $[2^m-1,m,2^{m-1}]$ simplex code has AWGNC, BSC,
  and max-fractional pseudocodeword redundancies
  \begin{eqnarray*}
    \rho(\code) \le \frac{(2^m-1)(2^{m-1}-1)}{3} \; . 
  \end{eqnarray*}
\end{proposition}

\begin{proof}
  For $m \ge 2$, consider the binary parity-check matrix $\cH$ whose
  rows are exactly the codewords of the dual code $\code\dual$ (in
  this case the $[2^m-1,2^m-1-m,3]$ Hamming code) with Hamming weight
  equal to $3$. This $\cH$ is the incidence matrix for a BIBD with
  parameters $(n,w_r,\lambda)$ = $(2^m-1, 3, 1)$.
  Theorem~\ref{thm:pseudoweight_bound} gives
  $\weight^{\min}_{\maxfrac} \ge 2^{m-1}$.

  Note that the number of codewords of weight $3$ in the
  $[2^m-1,2^m-1-m,3]$ Hamming code is $(2^m-1)(2^{m-1}-1)/3$. This is
  due to the fact that there exists a $3:1$ mapping from all vectors
  of length $2^m-1$ and weight $2$ onto the codewords of weight 3.

  Next, we justify the claim that $\cH$ is the parity-check matrix of
  $\code$.  A theorem of Simonis~\cite{Simonis} states that if there
  exists a linear $[n,k,D]$ code then there also exists a linear
  $[n,k,D]$ code whose codewords are spanned by the codewords of
  weight $D$.  Since the Hamming code is unique for the parameters
  $[2^m-1,2^m-1-m,3]$, this implies that the Hamming code itself is
  spanned by the codewords of weight $3$, so the rowspace of $\cH$
  equals $\code$.

  
  The result for AWGNC and BSC follows again by applying
  Theorem~\ref{thrm:pseudoredundancies}.
\end{proof}

We remark that the bounds of Propositions~\ref{prop:Hamming}
and~\ref{prop:Simplex} are sharp at least for the case $m=3$ and the
max-fractional weight, see Section~\ref{sec:exp-b}.

The following proposition proves that the AWGNC, BSC, and max-fractional
pseudocodeword redundancies are finite for all codes $\code$ with
minimum distance at most $3$.

\begin{proposition}\label{prop:D3_codes}
  Let $\code$ be a $[n,k,D]$ code with $D\leq 3$.  Then
  $\rho_{\maxfrac}(\code)$ is finite. Moreover, we have
  $\rho_{\maxfrac}(\code) = n-k$ in the case $D \leq 2$.
\end{proposition}

\begin{proof}
  First assume that $D\le 2$.  Let $\cH$ be any parity-check matrix
  for the code $\code$, let $\bldx$ be a nonzero pseudocodeword, and
  assume that $x_\ell$ is a maximal entry in $\bldx$ (for some $\ell
  \in \cI$).  We always have $\sum_{i\in\cI} x_i \ge x_\ell$ and
  hence $\weight_{\maxfrac} (\bldx) \geq 1$.
  
  Therefore, we may assume $D=2$.  Note that for such a code, $\cH$
  has no zero column and thus we may write by
  (\ref{eq:polytope-inequality-1})
  \[ x_\ell \le \sum_{i \in \cI \backslash \{ \ell \}} x_i \; ,
  \quad\text{ or }\quad 2 x_\ell \le \sum_{i \in \cI} x_i \; . \]

  From the definition of max-fractional weight, we obtain that
  $\weight_{\maxfrac} (\bldx) \ge 2$.  Choosing a parity-check matrix
  for $\code$ with $n-k$ rows we have that $\rho_{\maxfrac}(\code) =
  n-k$. From Theorem~\ref{thrm:pseudoredundancies},
  $\rho_{\AWGNC}(\code) = n-k$ and $\rho_{\BSC}(\code) = n-k$.

  Next, consider a code with minimum distance $D = 3$. Denote by $\cH$
  the parity-check matrix whose rows consist of \emph{all} codewords
  of the dual code of $\code$. Note that for a code of minimum
  distance $D$, a parity-check matrix $\cH$ consisting of all rows of
  the dual code $\code\dual$ is an orthogonal array of strength
  $D-1$. In the present case $D=3$, and this implies that in any pair
  of columns of $\cH$, all length-$2$ binary vectors occur with equal
  multiplicities (c.f. \cite[p. 139]{MacWilliams_Sloane}). Thus the
  matrix $\cH$ is an incidence matrix for a partial block design with
  parameters $(w_c,\lambda) = (2^{r-1},2^{r-2})$, where $r=n-k$.
  Therefore for this matrix $\cH$ the code has minimum (AWGNC, BSC, or
  max-fractional) pseudoweight at least $1 + w_c/\lambda = 3$, and it
  follows that the pseudocodeword redundancy is finite for any code
  with $D=3$.
\end{proof}



\section{Some Experimental Results}

\subsection{Cyclic Codes Meeting the Eigenvalue Bound}

We consider cyclic codes of length $n$ with full circulant
parity-check matrix $\cH$.  Thus $\cH=(H_{j,i})_{i,j \in \cI}$ is a
square matrix with entries $H_{j,i}=c_{j-i}$ for some vector $\bldc$
of length $n$, where all the indices are modulo $n$.  This $n\times n$
matrix is then $w$-regular (i.e.\ $(w,w)$-regular), where
$w=\sum_{i\in\cI}c_i$, so we may use the eigenvalue-based lower bound
in Proposition~\ref{prop:KV_bound} to examine the AWGNC pseudocodeword
redundancy.

For the largest eigenvalue of the matrix $\cL = \cH^T \cH$ we have
$\mu_1=w^2$, since every row weight of $\cL$ equals
$\sum_{i,j\in\cI}c_ic_j=w^2$.  Consequently, the eigenvalue bound is
\[ 
\weight_{\AWGNC}^{\min} \geq n\,\frac{2w-\mu_2}{w^2-\mu_2} \; ,
\] 
where $\mu_2$ is the second largest eigenvalue of $\cL$.


We carried out an exhaustive search on all cyclic codes $\code$ up to
length $n\leq 250$ and computed the eigenvalue bound in all cases
where the Tanner graph of the full circulant parity-check matrix is
connected.\footnote{This computation was done by a self-written C
  program.}  Table~\ref{table:KV_bound-1} gives a complete list of all
cases in which the eigenvalue bound equals the minimum Hamming
distance $D$ and $D\ge 3$.  In particular, the AWGNC pseudoweight
equals the minimum Hamming distance in these cases as well and thus we
have for the pseudocodeword redundancy $\rho_{\AWGNC}(\code)\leq n$.

\begin{table}
  \caption{Cyclic codes up to length 250 
    with $D\geq 3$ meeting the eigenvalue bound}
  \begin{tabular}{ccl}
    parameters & $w$-regular & comments \\\hline
    $[n,1,n]$ & $2$ & repetition code, $n=3\dots 250$ \\
    $[n,n-m,3]$ & $2^{m-1}$ & Hamming code, $n=2^m-1$, $m=3\dots 7$ \\
    $[7,3,4]$ & $3$ & dual of the $[7,4,3]$ Hamming code \\
    $[15,7,5]$ & $4$ & Euclidean geometry code EG(2,4) \\
    $[21,11,6]$ & $5$ & projective geometry code PG(2,4) \\
    $[63,37,9]$ & $8$ & Euclidean geometry code EG(2,8) \\
    $[73,45,10]$ & $9$ & projective geometry code PG(2,8) \\
  \end{tabular}
  \label{table:KV_bound-1}
\end{table}



\subsection{The Pseudocodeword Redundancy for Codes of Small
  Length}\label{sec:exp-b}

Let $\code$ be a binary linear code with parameters $[n,k,D]$ and let
$r=n-k$.  Two parity-check matrices $\cH$ and $\cH'$ of $\code$ are
called \emph{equivalent} if $\cH$ can be transformed into $\cH'$ by a
sequence of row and column permutations.  In this case,
$\weight^{\min}(\cH) = \weight^{\min}(\cH')$.



We computed the AWGNC, BSC, and max-fractional pseudocodeword
redundancies for all codes up to length 9.  Note that for this also
all possible parity-check matrices (up to equivalence) had to be
examined.\footnote{The enumeration of codes and parity-check matrices
  was done by self-written C programs.  The minimum pseudoweight for
  the various parity-check matrices was computed by using Maple 12 and
  the Convex package by Matthias Franz, available at
  http://www-fourier.ujf-grenoble.fr/\~{}franz/convex/.}  The main
observations are the following:

\begin{itemize}
\item If $D\geq 3$ then for \emph{every} parity-check matrix $\cH$ we
  have $\weight^{\min}_{\AWGNC}(\cH)\geq 3$.  This is not true for the BSC.
\item If $k=2$, then $\rho_{\AWGNC}(\code) = \rho_{\BSC}(\code) =
  \rho_{\maxfrac}(\code) = r$.
\item For the $[7,4,3]$ Hamming code $\code$ we have
  $\rho_{\AWGNC}(\code)=r=3$, $\rho_{\BSC}(\code)=4$, and
  $\rho_{\maxfrac}(\code)=7$.
\item For the $[7,3,4]$ simplex code $\code$ we have
  $\rho_{\AWGNC}(\code)=r=4$, $\rho_{\BSC}(\code)=5$, and
  $\rho_{\maxfrac}(\code)=7$.  There is (up to equivalence) only one
  parity-check matrix $\cH$ with $\weight_{\AWGNC}^{\min}(\cH)=4$,
  namely
  \[ \cH = \mat{
    \ze&\ze&\zo&\ze&\zo&\zo&\zo\\
    \zo&\ze&\ze&\zo&\ze&\zo&\zo\\
    \zo&\zo&\ze&\ze&\zo&\ze&\zo\\
    \zo&\zo&\zo&\ze&\ze&\zo&\ze\\
  }\;. \] It is the only parity-check matrix with constant row weight
  3.
\item For the $[8,4,4]$ extended Hamming code $\code$ we have
  $\rho_{\AWGNC}(\code)=5$, $\rho_{\BSC}(\code)=6$, and
  $\rho_{\maxfrac}(\code)=\infty$.  This code $\code$ is the shortest
  one such that $\rho_{\AWGNC}(\code)>r$, and also the shortest one
  such that $\rho_{\maxfrac}(\code)=\infty$.
\end{itemize}


\section*{Acknowledgements}

The authors would like to thank Nigel Boston, Christine Kelley, and
Pascal Vontobel for helpful discussions.  This work was supported in
part by the Science Foundation Ireland (Claude Shannon Institute for
Discrete Mathematics, Coding and Cryptography, Grant 06/MI/006, and
Principal Investigator Award, Grant 08/IN.1/I1950). The work of Vitaly
Skachek was done in part while he was with the Claude Shannon
Institute, University College Dublin. His work was also supported in
part by the National Research Foundation of Singapore (Research Grant
NRF-CRP2-2007-03).

\end{document}